\documentclass[conference,letterpaper]{IEEEtran}
\IEEEoverridecommandlockouts % Don't forget this command!

\usepackage[letterpaper,%
            left=.64in,right=.64in,top= 0.64in,bottom= 0.64in]{geometry}
\usepackage{amsfonts}                                                          
\usepackage{epsfig}

\usepackage{amsthm}
\usepackage{graphicx}        
\usepackage{latexsym}
\usepackage{amssymb}
\usepackage{amsmath}
\usepackage{parskip}
\usepackage{multirow}
\usepackage{cite}
\usepackage{mathtools}
\usepackage{epstopdf}
\usepackage{etoolbox}
\usepackage{array}
\usepackage{mathptmx}
\usepackage[bottom]{footmisc}
\usepackage{tikz}

\usetikzlibrary{decorations.pathreplacing}

\usepackage{verbatim}
\usepackage{calrsfs}
\usepackage{booktabs}

\DeclareMathAlphabet{\pazocal}{OMS}{zplm}{m}{n}

\let\bbordermatrix\bordermatrix
\patchcmd{\bbordermatrix}{8.75}{4.75}{}{}
\patchcmd{\bbordermatrix}{\left(}{\left[}{}{}
\patchcmd{\bbordermatrix}{\right)}{\right]}{}{}
\renewcommand{\theequation}{\thesection.\arabic{equation}}
\renewcommand{\thefigure}{\thesection.\arabic{figure}}
\renewcommand{\thetable}{\thesection.\arabic{table}}

\newcommand{\sr}{\stackrel}

\newcommand{\rar}{\rightarrow}

\newcommand{\tri}{\sr{\triangle}{=}}

\newcommand{\be}{\begin{equation}}
\newcommand{\ee}{\end{equation}}
\newcommand{\bea}{\begin{eqnarray}}
\newcommand{\eea}{\end{eqnarray}}
\newcommand{\bes}{\begin{eqnarray*}}
\newcommand{\ees}{\end{eqnarray*}}
\newcommand{\bce}{\begin{center}}
\newcommand{\ece}{\end{center}}
\newcommand{\beae}{\begin{IEEEeqnarray}{rCl}}
\newcommand{\eeae}{\end{IEEEeqnarray}}

\def\VR{\kern-\arraycolsep\strut\vrule &\kern-\arraycolsep}
\def\vr{\kern-\arraycolsep & \kern-\arraycolsep}

\newcommand{\ben}{\begin{enumerate}}
\newcommand{\een}{\end{enumerate}}

\newcommand{\hso}{\hspace{.1in}}
\newcommand{\hst}{\hspace{.2in}}

\newcommand{\noi}{\noindent}

\newtheorem{theorem}{Theorem}[section]
\newtheorem{problem}{Problem}[section]
\newtheorem{remark}{Remark}[section]

\newtheorem{definition}{Definition}[section]

\newtheorem{innercustomgeneric}{\customgenericname}
\providecommand{\customgenericname}{}
\newcommand{\newcustomtheorem}[2]{%
  \newenvironment{#1}[1]
  {%
   \renewcommand\customgenericname{#2}%
   \renewcommand\theinnercustomgeneric{##1}%
   \innercustomgeneric
  }
  {\endinnercustomgeneric}
}

\newcustomtheorem{customcond}{Condition}
\newcustomtheorem{customlemma}{Lemma}

\usepackage{lipsum}

\begin{document}

%\baselineskip=16pt
%\sloppy

%% Paper Title
%% You can use linebreaks \\ within to get better formatting as
%% desired.

\title{Feedback  Capacity Formulas of AGN  Channels Driven by  Nonstationary   Autoregressive Moving Average Noise}

\author{%
   \IEEEauthorblockN{Stelios Louka,    
   Christos Kourtellaris and Charalambos D. Charalambous                  }
   \IEEEauthorblockA{%
                     Department of Electrical and Computer Engineering\\University of Cyprus, 75 Kallipoleos Avenue, P.O. Box 20537, Nicosia, 1678, Cyprus, \\ 
 \{slouka01,kourtellaris.christos,chadcha@ucy.ac.cy\}}}

\maketitle

\begin{abstract}
%Current telecommunication and information systems are designed based on Shannon's operational definitions of coding-capacity for reliable communication, which utilizes encoders and decoders,  to combat communication noise and to remove redundancy in data. 
 In this paper we derive closed-form formulas  of feedback capacity  and nonfeedback achievable rates,  for  Additive Gaussian Noise (AGN) channels driven by nonstationary autoregressive moving average (ARMA) noise (with unstable one poles and zeros),  based on  time-invariant feedback codes and channel input distributions. 
From the analysis and simulations follows the  surprising observations, (i) the use of time-invariant channel input distributions gives rise to multiple regimes of capacity that depend on  the   parameters of the ARMA noise, which may or may not use feedback,  (ii) the more unstable the pole (resp. zero) of the ARMA noise the higher (resp. lower) the feedback capacity,  (iii) certain conditions, known as detectability and stabilizability  are necessary and sufficient to ensure the feedback capacity formulas and nonfeedback achievable rates {\it are independent of the initial state of the ARMA noise}.   Another surprizing observation is that  Kim's \cite{kim2010} characterization of feedback capacity which is developed for  stable ARMA noise, if applied to the  unstable ARMA noise, gives a lower value of feedback capacity compared to our feedback capacity formula. 
\end{abstract}

%\newpage

%\tableofcontents
%\newpage

\section{Introduction}
\label{sect:intro}
The AGN channel is defined by 
\begin{align}
&Y_t=X_t+V_t, \hso t=1, \ldots, n, \hso  \frac{1}{n}  {\bf E} \Big\{\sum_{t=1}^{n} (X_t)^2\Big\} \leq \kappa \label{AGN_in}
\end{align}
%subject to the average power constraint
%\begin{align}
%&\frac{1}{n}  {\bf E}_{v_0} \Big\{\sum_{t=1}^{n} (X_t)^2\Big\} \leq \kappa, \hso \kappa \in [0,\infty), \label{AGN_innew2}
%\end{align}
where   $\kappa \in [0,\infty)$ is the total  power of the transmitter,      $X^n = \{X_1, X_2, \ldots, X_n\}$,   $Y^n = \{Y_1, Y_2, \ldots, Y_n \}$ and $V^n = \{ V_1, \ldots, V_n\}$,    are   the sequences of channel input, channel   output, and Gaussian noise  random variables (RVs), respectively. \\
The feedback and nonfeedback capacity of the AGN channel, when the noise $V^n$ is stable, stationary, or asymptotically stationary,  can be considered to have been explained sufficiently in information theory \cite{gallager1968,cover-pombra1989}. The most general  is the Cover and Pombra  formulation and coding theorems \cite[Theorem~1]{cover-pombra1989},  for  the set of uniformly distributed messages $W : \Omega \rar  {\cal M}^{(n)} \tri  \left\{1, 2,\ldots, \lceil M_n \rceil\right\}$, codewords  of block length $n$, $X_1=e_1(W),\ldots,  X_n=e_n(W,X^{n-1}, Y^{n-1})$,   decoder functions,  $y^n \longmapsto d_{n}(y^n)\in  {\cal M}^{(n)}$, with average  error  probability
\bea
{\bf P}_{error}^{(n)} =  \frac{1}{\lceil M_n \rceil} \sum_{w=1}^{\lceil M_n \rceil} {\bf  P}\Big(d_n(Y^n) \neq W\Big|W=w\Big). \label{g_cp_4}
\eea

 The objective of this paper is twofold. 

1)  To show that feedback and nonfeedback capacity formulas and achievable rates, may behave very different, depending on the definitions of achievable rates, in particular, whether conditions are imposed to ensure these rates are insensitive to  initial states or distributions of the channel, i.e., of  $V^n$.

2) To show, the  surprizing result that,  the consideration of an unstable noise $V^n$ alters  the mathematical formulas of  feedback and nonfeedback capacity formulas and achievable rates, and that noises with unstable poles  give significant gains of achievable rates, at no extra expense of power. 

To keep the analysis simple,  we consider      the {\it unstable and stable},  autoregressive moving average,  unit memory noise, denoted by ARMA$(a,c), a \in (-\infty,\infty), c\in (-\infty,\infty),  c \neq a$, as defined    below. Versions of {\it stable or marginally stable},  ARMA$(a,c), a \in [-1,1], c\in [-1,1]$ noise are considered  since the early 1970's, in \cite{butman1969,tienan-schalkwijk1974,wolfowitz1975,butman1976,ozarow1990,yang-kavcic-tatikonda2007,kim2010},   where the reader may find bounds on achievable rates of feedback and nonfeedback codes, under various assumptions and formulations. 
\begin{align}
 &\mbox{ARMA$(a,c)$:} \hso  V_t=c V_{t-1}+ W_t - aW_{t-1}, \hso  t=1, \ldots, n, \label{ar_1}  \\
 & W_t \in N(0,K_{W}),\: K_W>0, \:   t=1, \ldots, n,\: \mbox{ mutually indep}, \label{ar_3}\\
 & \{W_1, \ldots, W_n\} \; \mbox{indep. of initial state $S\tri (V_0,W_0)$}, \\
 &V_0\in N(0, K_{V_0}), \hso   W_0\in N(0, K_{V_0}), \hso K_{V_0}\geq 0, \hso K_{W_0} \geq 0, \\
 &   a \in (-\infty,\infty), \hso c\in (-\infty,\infty), \hso c \neq a, \label{ar_4}
 \end{align}
where  the notation,  $Z\in N(0,K_{Z})$ means  $Z$  is a Gaussian RV,  with  zero mean, and  variance $K_Z$. The  $ARMA(a,c)$ noise  is  equivalently expressed in state form, with state  $S_t$ as,  
\begin{align}
&S_t \tri \frac{cV_{t-1}-aW_{t-1}}{c-a}, \hso t =1, \dots,n, \\
&S_{t+1} = cS_t+W_t, \hso S_1=s,  \hso t =1, \dots,n,\\
&V_t = (c-a)S_t +W_t, \hso t =1, \dots,n
\end{align}  
From the  $ARMA(a,c)$ noise, follow the  two  special cases, 
\begin{align}
&\mbox{Autoregressive: }AR(c)\big|_{a=0}:\hso V_t = cV_{t-1}+W_t \label{AR}\\
&\mbox{Moving Average: }MA(a)\big|_{c=0}:\hso V_t = W_t-aW_{t-1} \label{MA}
\end{align}

%According to the literature, such as, the excellent manuscript of Gallager \cite{gallager1968}, one approach to address the issue  of ``sensitivity of achievable rates with respect to initial channel condition'',   is  to take the supremum over all unknown channel parameters of the error exponents or achievable rates. Another approach, which we follow in this paper is to formulate  problems achievable rates with and without feedback, by imposing the mathematical restrictions such that the rates do not depend on the channel parameters.

  \subsection{Literature Review } 
Due to the relevance to our investigation,  of prior formulas  found in  \cite{butman1969,tienan-schalkwijk1974,wolfowitz1975,butman1976,cover-pombra1989,ozarow1990,yang-kavcic-tatikonda2007,kim2010}, we briefly discuss some of these below, with emphasis on the  {\it formulations and  assumptions}.

3) {\it Formulation and Bounds  with Initial State \cite{tienan-schalkwijk1974,wolfowitz1975,butman1976}.} Wolfowitz \cite{wolfowitz1975} and Butman  \cite{butman1976}, derived a  lower bound on feedback capacity of the AGN channel, driven  by the noise,  AR$(c), c\in [-1,1]$, using the  linear feedback coding scheme, $X_t=g_t\Big(\Theta-{\bf E}\big\{\Theta\Big|Y^{t-1},v_0\big\}\Big), t=2, \ldots, n,  X_1=g_1\Theta$, where $g_t$ are nonrandom real numbers, and  $\Theta: \Omega \rar {\mathbb R}$ is Gaussian, $\Theta \in N(0,1)$, i.e.,  under the assumption the initial state $V_0=v_0$ is known to the encoder and decoder. The lower bound   is 
\begin{align}
&C^{LB}  
 =\frac{1}{2}\log \chi^2, \hso \mbox{where $\chi$ is the positive root of}  \label{butman_cs_2}\\
&\chi^4-\chi^2 - \frac{\kappa}{K_W} \Big(\chi+|c|\Big)^2=0, \; |c|\leq 1, \; K_W>0. \label{butman_cs_3}
\end{align}
Butman \cite{butman1976}   conjectured that $C^{LB}$ is the feedback capacity.\\ \cite{butman1976}  includes a comparison to the Tienan and Schalkwijk upper bound \cite{tienan-schalkwijk1974}, also derived under the same assumptions. 

\begin{figure}
\centering
  \includegraphics[width=\linewidth]{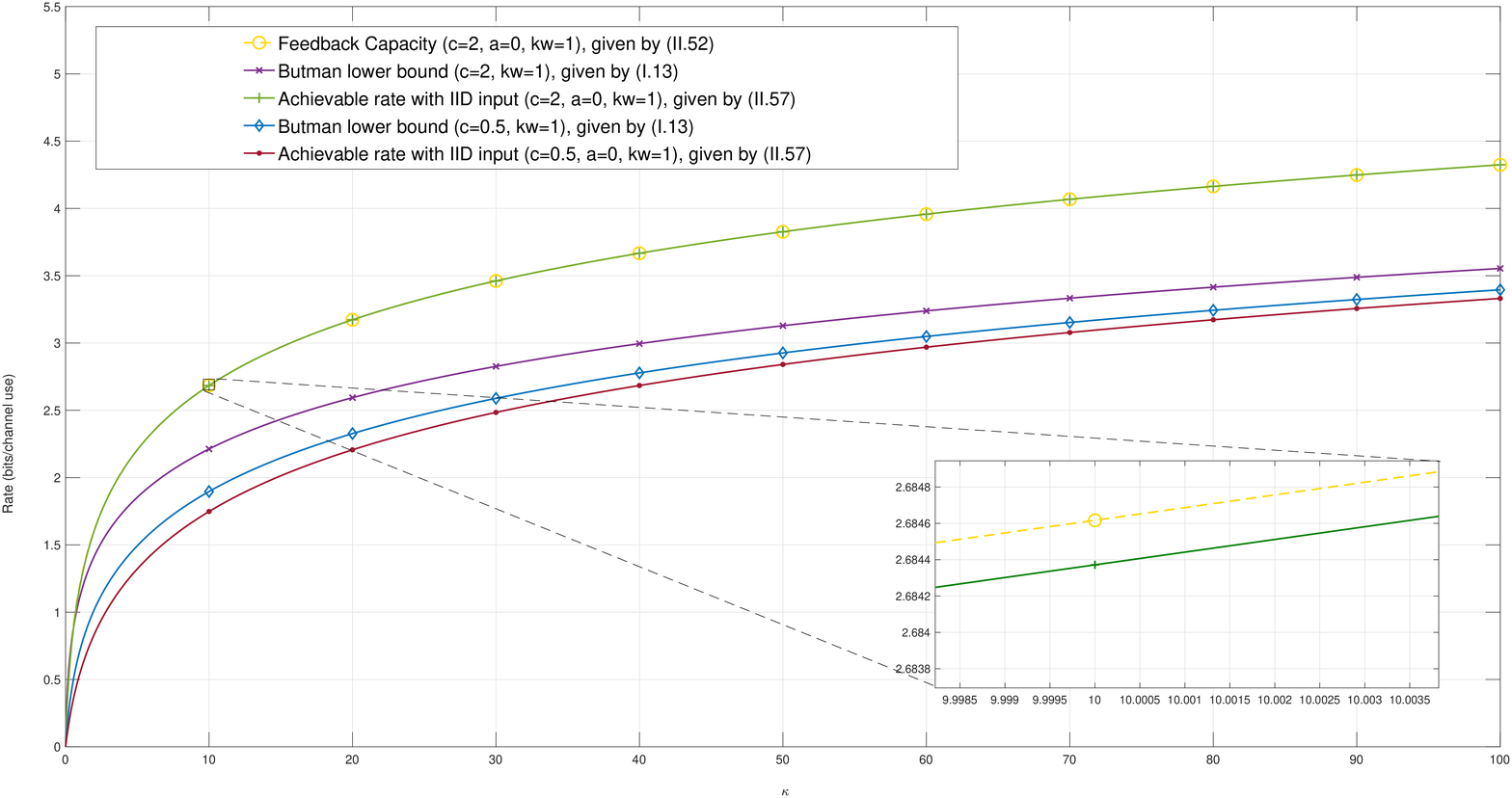}
   \vspace{-0.9cm}
  \caption{Comparison of $C_{FB}=C^{YKT}=C^{LB}$, i.e.,   Butman's  (\ref{butman_cs_2}), Kim's feedback capacity \cite[Theorem~6.1, $C_{FB}$]{kim2010},   and our feedback capacity and achievable nonfeedback rate  for unstable and stable AR$(c)$ noise}
  \label{figButman}
\end{figure}

{\it 4) Formulation,   Characterizations  and Closed-Form Formulas with and without Initial State  \cite{cover-pombra1989,yang-kavcic-tatikonda2007,kim2010}.}\\
{\it 4.1) Cover and Pombra \cite{cover-pombra1989}} characterized the feedback (and nonfeedback)  capacity of the AGN channel, driven by a nonstationary Gaussian noise $V^n$,   \cite[Theorem~1]{cover-pombra1989}, for codes that do not assume knowledge of the initial state of the noise. The feedback capacity is  $C \tri \lim_{n \longrightarrow \infty}  \frac{1}{n} C_n$ (provided the limit exists), where 
 $ C_n$ is the $n$-finite block  length or transmission feedback information ($n-$FTFI) capacity,  
 \begin{align}
 &C_n
  \tri  \sup_{ \frac{1}{n}   {\bf E} \big\{\sum_{t=1}^n (X_t)^2 \big\}\leq \kappa}  H(Y^n) - H(V^n )  \label{cp_11}\\
&  X_n=\sum_{j=1}^{n-1} \Gamma_{n,j}V_j +\overline{Z}_n, \hso  n=1, 2, \ldots, \hso X_1=\overline{Z}_1,  \label{cp_in_1} 
  \end{align}
  and the supremum is over  nonrandom  $\Gamma_{n,j}, j=1, \ldots, n-1$ and jointly correlated, Gaussian RVs $\{\overline{Z}_1, \ldots, \overline{Z}_n\}$, independent of $V^n$; $H(X)$ is the differential entropy of RV $X$.

\noi {\it 4.2) Yang, Kavcic and Tatikonda \cite{yang-kavcic-tatikonda2007}}  analyzed the feedback capacity of the AGN channel driven a noise $V^n$, with state $S^n$, under the following assumption.

{\it Assumption (YKT) \cite[page 933, I)-III)]{yang-kavcic-tatikonda2007}:   given the initial state of the noise $S_1=s$, which is known to the encoder and the decoder, the channel input $X^n\tri \{X_1, \ldots, X_n\}$ uniquely defines the state variables $S^n$ and vice-versa. }

\noi \cite[Theorem~7]{yang-kavcic-tatikonda2007},  computed  the feedback rate of the AGN channel driven by   ARMA$(a,c), a \in (-1,1), c\in (-1,1)$ noise,  using  the definition,   
\begin{align}
&C^{YKT} \tri  \sup_{(\Lambda, K_Z):   \lim_{n \longrightarrow \infty} \frac{1}{n} {\bf E}\big\{\sum_{t=1}^n \big(X_t\big)^2\big|S_1=s \big\}\leq \kappa} \lim_{n \longrightarrow \infty} \Big\{  \nonumber \\
&\hspace*{2.0cm}   \frac{1}{n} \sum_{t=1}^n H(Y_t|Y^{t-1}, s) - H(V^n|s)\Big\}, \label{ykt_ff}\\
&  X_n = \Lambda \Big(S_n - {\bf E}\Big\{S_n\Big|Y^{n-1}, S_1=s\Big\}\Big)+Z_n, \; n=2, \ldots, \label{YKT_input_ss_a}\\
&X_1=Z_1, \hso  Z_n \in N(0,K_{Z}),  \;  K_Z \geq0,  \hso   n=1,2, \ldots      \label{YKT_input_ss}
 \end{align} 
 where the RVs,  $\{Z_1, Z_2, \ldots, Z_n\}$ are mutually independent. 
The limiting  problem in the steady state,  is solved,  and 
 a formula of  $C^{YKT}$ is obtained. 
 % which  is a generalization of Butman's lower bound (\ref{butman_cs_2}), (\ref{butman_cs_3}).      \\
For the AR$(c), c\in (-1,1)$ noise,  $C^{YKT}=\frac{1}{2}\log |\Lambda^*|^2$, where  $\Lambda^*$  satisfies Butman's equation (\ref{butman_cs_3}),  i.e., $C^{YKT}=C^{LB}$ of Butman  (see \cite[Corollary~7.1]{yang-kavcic-tatikonda2007}).

\begin{remark}
\label{rem:zi}
In both \cite[Theorem~7,  Corollary~7.1]{yang-kavcic-tatikonda2007},     $C^{YKT}$, are  achieved  by $\Lambda=\Lambda^*$ and $K_Z=K_{Z}^*=0$.
%  if $a=0$ then  $\Lambda^*$ satisfies Butman's eqn (\ref{butman_cs_3}).  However, as easily verified from  (\ref{kim_in_2}), if $K_Z=K_Z^*=0$ then $H(Y_t|Y^{t-1}, s) - H(V^n|s)=0, \forall n$, hence $C^{YKT}=0$. It appears this contradiction is recognized in \cite{yang-kavcic-tatikonda2007} and discussed in the conclusion.  
\end{remark}

{\it 4.3) Kim \cite{kim2010} re-visited}  the AGN channel driven by a stable noise, and derived characterizations of feedback capacity,  in frequency domain \cite[Theorem~4.1]{kim2010}, with zero power spectral density  of the innovations  part of the channel input $X_n, n=1,2, \ldots$, and in  time-domain  \cite[Theorem~6.1]{kim2010}, with zero variance  of the innovations  part, $Z_n$  of the  input $X_n, n=1,2, \ldots$,  (\ref{YKT_input_ss_a}), (\ref{YKT_input_ss}).  Kim \cite[page 78, first paragraph]{kim2010}, stated  that \cite[Theorem~6.1, $C_{FB}$]{kim2010}) characterizes the feedback capacity (i.e., $C_{FB}$) as conjectured in  \cite[Theorem~6 and Conjecture~1]{yang-kavcic-tatikonda2007}, in  a simpler form (without the  innovations process of the input).\\
In  \cite[Page 76]{kim2010}, the AGN channel  with ARMA$(a,c), a \in (=[-1,1], c\in [-1,1]$ noise is considered, with initial state  $(V_0=0,W_0=0)$, and the  feedback capacity $C_{FB}$ is computed based on the limiting expression  (\ref{ykt_ff}), with   input  $X_n$, replaced by 
\begin{align}
&X_1= Z_1, \hst \mbox{$Z_n$ is zero mean, variance  $K_{Z}>0$}
 \label{kim_in_2}\\
&X_n=\Lambda  \Big( S_n-{\bf E}\Big\{S_n \Big|Y^{n-1}\Big\}\Big), \hst n=2,  \ldots  \label{kim_in_3}
\end{align}
%This calculation gave an expression of $C_{FB}$ which is precisely the maximal information rate $C^{YKT}$ derived in Yang, Kavcic and Tatikonda \cite[Theorem~7]{yang-kavcic-tatikonda2007}. 
 For the AR$(c), c \in (-1,1)$ noise, the optimal  $\Lambda=\Lambda^*$  satisfies Butman's equation (\ref{butman_cs_3}), and  $C_{FB}=\frac{1}{2}\log |\Lambda^*|^2$, i.e.,  identical to Butman's lower bound.

\subsection{Main Problem of the Paper: Brief Discussion of Results and Comparisons}
\label{sect:results}
As we briefly demonstrate via simulations of our feedback capacity expressions, which we derived in Section~\ref{sec:th},  that even for stable ARMA$(a,c), a \in (-1,1), c\in (-1,1)$ noise,  we arrive at completely different formulas, compared to those in \cite{yang-kavcic-tatikonda2007,kim2010}.  We show these  differences are attributed to our feedback capacity problem definition, stated as Problem~\ref{problem_1}. In particular,   the limiting expression of feedback rate,   (\ref{inter}),  is independent of the initial state of the noise, i.e., $C^{\infty}(\kappa,s)=C^{\infty}(\kappa), \forall s$, when    compared to (\ref{ykt_ff}).  

\begin{problem} 
\label{problem_1}
Consider  the AGN channel  driven by the   $ARMA(a,c)$ noise, with  $c\in (-\infty,\infty)$,  $a\in (-\infty,\infty)$, and with initial state $S_1=s$, known to the encoder and the decoder.  Define the feedback rate
\begin{align}
C^{\infty}(\kappa,s) \tri& \sup_{ \lim_{n \longrightarrow \infty} \frac{1}{n} {\bf E}\big\{\sum_{t=1}^n \big(X_t\big)^2\big|s \big\}\leq \kappa} \lim_{n \longrightarrow \infty} \Big\{ \frac{1}{n} \sum_{t=1}^n H(Y_t|Y^{t-1}, s) \nonumber \\
&- H(V^n|s)\Big\} \label{inter}
 \end{align}
 where the supremum is taken over all jointly Gaussian    channel input process $X_n, n=1,2,\ldots$, generated by time-invariant feedback strategies, and   induce distributions  ${\bf P}_{X_t|X^{t-1}, Y^{t-1}, S_1}, t=1,2,  \ldots$, such that  the joint process $(X_n,Y_n), n=1,2,\ldots,$  is jointly Gaussian, for $S_1=s$.\\
We address the following questions.\\
(a) What are  necessary and/or sufficient conditions for, \\
(i)  asymptotic stationarity of the process $(X^{n}, Y^{n}), n=1,2,\ldots$ or  the marginal process $X^{n}$, that achieve $C^\infty(\kappa,s)$, and \\
(ii)  $C^\infty(\kappa,s)=C^{\infty}(\kappa)\;  \forall s$, i.e.,  independent of initial data?\\
(b) What are the  closed form formulas of  feedback capacity $C^\infty(\kappa,s)=C^{\infty}(\kappa)\; \forall s$?
 \end{problem}

Problem~\ref{problem_1}.(a).(i),(ii), captures the requirement that $C^{\infty}(\kappa,s)$ is well-defined,  for unstable $ARMA(a,c), c\in (-\infty,\infty),  a\in (-\infty,\infty)$ noise (as well as stable,   $c\in (-1,1), a\in (-1,1)$), and $C^{\infty}(\kappa,s)=C^{\infty}(\kappa), \forall s$. \\
In the rest of the paper,  we show  there are  multiple regimes of $C^{\infty}(\kappa)$, which depend on the parameters  $(a,c,\kappa)$. At some  regimes,  feedback does not increase the capacity. This is attributed to the use of time-invariant channel input strategies. For these regimes  we derive achievable nonfeedback rates $C^{\infty,nfb}_{LB}(\kappa)$, based on a simple IID  channel input $X_n=Z_n, Z_n \in N(0, \kappa), n=1,2, \ldots$.\\
Fig.~\ref{figButman}, compares our feedback capacity and achievable nonfeedback rate to Butman's lower bound and Kim's feedback capacity \cite[Theorem~6.1, $C_{FB}$]{kim2010}, for the AR$(c)$, stable and unstable  noise. For unstable noise, even an IID channel input outperforms the feedback rate  $C_{FB}=C^{YKT}=C^{LB}$.\\
%The question of Problem~\ref{problem_1}.(a).(ii), captures the requirement that $C^{\infty}(\kappa,s)=C^{\infty}(\kappa), \forall s$, i.e., the rate does not depend on the initial state $S_1=s$ or the distribution of ${\bf P}_{S_1}$. 
The verification of our objectives  described in Section~\ref{sect:intro}, under 1) and 2), are demonstrated in  Figures~\ref{fig1}, \ref{fig2}.
%which  presents evaluations of our closed-form formulas of feedback capacity of Problem~\ref{problem_1}.(b), i.e, 
%$C^{\infty}(\kappa,s)=C^{\infty}(\kappa), \forall s$, and achievable nonfeedback rates $C^{\infty,nfb}_{LB}(\kappa)$. 
\begin{figure}
\center
  \includegraphics[width=\linewidth]{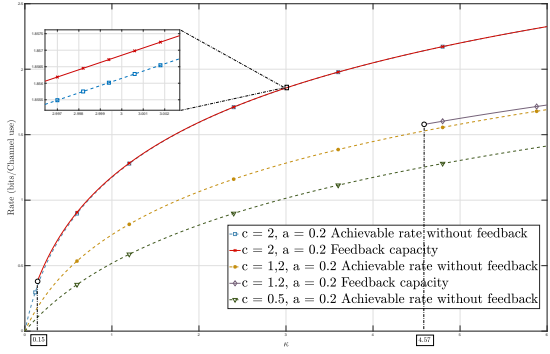}
  \vspace{-0.9cm}
  \caption{Feedback capacity $C^\infty(\kappa)$ for $\kappa \in {\cal K}^{\infty}(a,c,K_W)$ based on (\ref{sol_12}) and lower bound on nofeedback capacity $C_{LB}^{\infty,nfb}(\kappa)$ for  $\kappa \in (0,\infty)$ based on (\ref{nf_888}),    of the AGN channel driven by stable/unstable ARMA$(a,c)$ noise, for various values of $a=0.2$, $c\in (-\infty, \infty)$ and $K_W=1$.}
  \label{fig1}
  \vspace{-0.35cm}
\end{figure}
\begin{figure}
\centering
  \includegraphics[width=\linewidth]{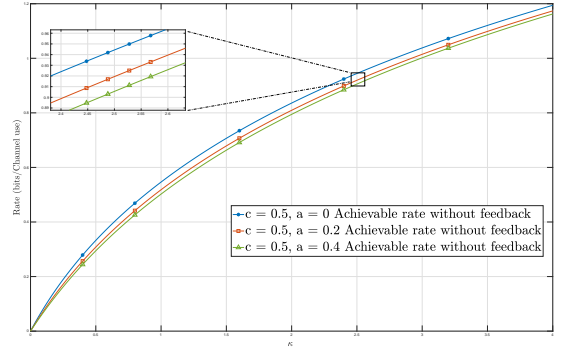}
   \vspace{-0.9cm}
  \caption{Lower bound achievable rate for no-feedback $C_{LB}^{\infty,nfb}(\kappa)$ for  $\kappa \in (0,\infty)$ based on (\ref{nf_888}), of the AGN channel driven by stable ARMA$(a,c)$ noise, for various values of $|a|\leq 1$, $c=0.5$ and $K_W=1$.}
  \label{fig2}
\end{figure}
%\begin{figure}
%\centering
%  \includegraphics[width=\linewidth]{figure3.eps}
%   \vspace{-0.9cm}
%  \caption{Comparison achievable rates between stable $AR(c)$ noise \cite{ckl2019} and stable $MA(a)$ noise, (\ref{nfma}) for various values of $a\in{0, 0.2}$, $c\in {0.2, 0}$ and $K_W=1$.}
%  \label{fig3}
%\end{figure}
Figure~\ref{fig1}, shows that  feedback capacity $C^{\infty}(\kappa)$ over the appropriate region, is an increasing function of the parameter $c$, i.e., the higher unstable pole, the higher the value of $C^{\infty}(\kappa)$. On the other hand, the higher the  value of $a$, i.e., of the zero,  the lower the value  $C^{\infty}(\kappa)$.\\
Figure~\ref{fig2} shows that the   lower bound on nofeedback capacity $C^{\infty,nfb}_{LB}(\kappa)$ is achievable for all $\kappa \in (0,\infty)$, for stable and unstable $ARMA(a,c)$ noise, and it is very closed  to the feedback capacity $C^{\infty}(\kappa)$.\\

\vspace{-0.5cm}
\section{Characterization of Feedback Capacity}\label{sec:th}
We consider a the set of uniformly distributed messages $W : \Omega \rar  {\cal M}^{(n)} \tri  \left\{1, 2,\ldots, \lceil M_n \rceil\right\}$,  codewords  $X_1=e_1(W,S_1),\ldots,  X_n=e_n(W,S_1, X^{n-1}, Y^{n-1})$,   decoder functions  $(s,y^n) \longmapsto d_{n}(s,y^n)\in  {\cal M}^{(n)}$, and  average  error  probability (\ref{g_cp_4}), which  is also conditioned on $S_1=s$.

\subsection{$n-$FTFI Capacity of  Time-Invariant Channel Input Strategies} 
 Since  our  code depends on the initial state of  ARMA$(a,c)$ noise, $S_1=s$, then  the entropies in  (\ref{cp_11}), (\ref{cp_in_1}), are conditional on $S_1=s$,  and $C_n^{fb}$ is denoted by  $C_n(\kappa,s)$,   and $X_n$ depends on $S_1=s$. 
%By \cite{paper-submitted} the optimal time-varying channel input distribution with feedback, for the optimization problem $C_n(\kappa,s)$,   is  conditionally Gaussian, of  the form
%\begin{align}
%{\bf P}_{X_t|X^{t-1},Y^{t-1},S_1}=&{\bf P}_{X_t|V^{t-1},Y^{t-1},S_1} \hst \mbox{by $Y_t=X_t+V_t$}\\
%=&{\bf P}_{X_t|S^{t},Y^{t-1},S_1}   \mbox{by $S_1, V^{t-1}$ uniquely species $S^t$} \nonumber  \\
%=&{\bf P}_{X_t|S_t,Y^{t-1},S_1}, \hst t=1, \ldots, n \label{tv_is}
%\end{align}
%This follows from 
Theorem~\ref{thm_FTFIi} is easily derived from the Cover and Pombra characterization (see \cite{charalambous-kourtellaris-louka:ARXIV-2020}) or  \cite{yang-kavcic-tatikonda2007}.

\begin{theorem}  Characterization of  $n-$FTFI Capacity \\
\label{thm_FTFIi}
Consider the AGN Channel driven by  ARMA$(a,c)$ noise, and  a feedback code with knowledge of the noise initial state $S_1=s$, and let $\widehat{S}_{t} \tri {\bf E}\big\{S_{t} | Y^{t-1},S_1=s\big\}$, 
$K_{t}\tri {\bf E}\big\{\big(S_{t} - \widehat{S}_{t}\big)^2|S_1=s \big\},  K_1 = 0,  t=2, \ldots, n.$  The analog of (\ref{cp_11}) and (\ref{cp_in_1}) are given as follows. 
\begin{align}
&X_t =\Lambda_t \Big({S}_{t} - \widehat{S}_{t}\Big) + Z_t,  \hso t=1,\dots,n, \label{Q_1_3_s1_a} \hso X_1 =  Z_1, \\
&Z_t\in  N(0, K_{Z_t}),  \; \mbox{a  Gaussian sequence,} \; t=1, \ldots, n, \label{Q_1_5_s1} \\
&Z_t \hso  \mbox{independent of}  \hso  (V^{t-1},X^{t-1},Y^{t-1}, {S}_1), \label{Q_1_6_s1}\\
&V_t=(c-a) S_{t} +W_t, \hso S_1=s,\hso t=1, \ldots, n,\label{Q_1_3_a_s1_a} \\
&Y_t= X_t + V_t=  {\Lambda}_t \Big({S}_{t} - \widehat{S}_{t}\Big)+Z_t+V_t \label{Q_1_3_a_s1} \\& \hso ={\Lambda}_t \Big({S}_{t} - \widehat{S}_{t}\Big)+\Big(c-a\Big)S_t+W_t+Z_t,  \\ 
&Y_1=Z_1+\Big(c - a\Big) S_1 +W_1, \hso S_1=s, \label{Q_1_4_s1_a}\\
&S_{t+1} = cS_t +W_t \label{stateeq} \hso S_1=s, \hso t = 2, \dots, n,
\end{align}
\begin{align}
&\frac{1}{n}  {\bf E} \Big\{\sum_{t=1}^{n} \big(X_t\big)^2\Big| S_1=s\Big\}=\frac{1}{n}   \sum_{t=1}^n \Big\{\big(\Lambda_t\big)^2  K_{t}  + K_{Z_t} \Big\},     \label{cp_e_ar2_s1}\\
&(\Lambda_t, K_{Z_t})\in (-\infty, \infty) \times [0,\infty) \hso \mbox{scalar,  non-random,}
%\\
%&\widehat{S}_{t} \tri {\bf E}_{s}\Big\{S_{t} \Big| %Y^{t-1},S_1=s\Big\} \\
%& K_{t}\tri {\bf E}_{s}\left\{\Big(S_{t} - \widehat{S}_{t}%\Big)^2  \right\}, \hso K_1 = 0, \hst  t=2, \ldots, n.%\label{s1_24}
\end{align}  
Further, $H(Y^n|s)-H(V^n|s)$,     $(\widehat{S}_{t}, K_{t}), t=1, \ldots, n$ are determined  by the generalizedl  time-varying Kalman-filter.\\
%   of estimating $S^n$ from $Y^n$, given below.\\
{\it Kalman-filter recursion:} 
\begin{align}
&\widehat{S}_{t+1} = c \widehat{S}_{t} + M_t(K_{t}, \Lambda_t, K_{Z_t}) I_t, \hso \widehat{S}_{1}=s, \label{Q_1_8_s1} \\
&I_t \tri Y_t-{\bf E}_{s}\Big\{Y_t\Big|Y^{t-1}\Big\}= Y_t - \Big(c-a \Big) \widehat{S}_{t}, \nonumber \\& I_1= Z_1+W_1,\hso    t=1, \ldots, n, \label{Q_1_9_s1}\\
& \hso = \Big(\Lambda_t +c-a\Big)\Big(S_{t}- \widehat{S}_{t}\Big) + Z_t +W_t,\label{Q_1_9_s1_n} \\
&M_t(K_{t}, \Lambda_t, K_{Z_t})  \tri  \Big( K_{W_t} + c  K_{t}\Big(\Lambda_t + c -a\Big)\Big)\cdot \nonumber \\ & \hst \Big(K_{Z_t}+ K_{W_t} + \Big(\Lambda_t + c -a\Big)^2 K_{t}\Big)^{-1}\label{Q_1_11_s1}
 \end{align}
{\it Error recursion, $E_t\tri S_t-\widehat{S}_t$ :}
 \begin{align} 
&E_{t+1}=F_t(K_{t},{\Lambda}_t,K_{Z_t}) E_{t}-M_t(K_{t},{\Lambda}_t, K_{Z_t})  \Big(Z_t+ W_t\Big)+W_t, \nonumber \\ &  E_1=S_1-\widehat{S_1}=0, \ \ t=2, \ldots, n. \label{i_error_s1} \\
&F_t(K_t, {\Lambda}_t, K_{Z_t}) \tri c - M_t(K_t, {\Lambda}_t, K_{Z_t}) \Big({\Lambda}_t+c-a\Big)
\end{align} 
Entropy of channel output Process:
\begin{align}
&H(Y^n|s)
=\sum_{t=1}^n H(I_t), \hso I_t \hso \mbox{indep. innovation process.}   \label{Q_1_4_s1_aaa}
\end{align}
{\it Generalized time-varying difference Riccati  equation (DRE):}
\beae
K_{t+1}= && c^2 K_{t}  + K_{W} -\frac{ \Big( K_{W} + c K_{t}\big(\Lambda_t + c-a \big)\Big)^2}{ \Big(K_{Z_t}+ K_{W} + \big(\Lambda_t + c -a\big)^2 K_{t}\Big)},  \nonumber\\ &&K_t \geq0, \hso K_1=0, \hso t=2,\dots,n 
 \label{DRE_TV}
\eeae
The characterization of the $n-$FTFI capacity $C_n(\kappa,s)$  is
\begin{align}
C_{n}&(\kappa,s)=
 \sup_{\big(\Lambda_t, K_{Z_t} \big), t=1,\ldots, n: \hso \frac{1}{n} \sum_{t=1}^n \big\{\big(\Lambda_t\big)^2 K_{t}+ K_{Z_t}\big\}   \leq \kappa}\Big\{   \nonumber  \\&\frac{1}{2} \sum_{t=1}^n   \log\Big( \frac{\big(\Lambda_t+c-a\big)^2 K_{t}  + K_{Z_t} +K_{W}}{K_{W}}\Big) \Big\}. \label{IH}
\end{align} 
\end{theorem}

\begin{proof} See \cite{charalambous-kourtellaris-louka:ARXIV-2020}.
\end{proof}

Unlike the Cover and Pombra  problem (\ref{cp_11}) and (\ref{cp_in_1}),  problem (\ref{IH})  is sequential, hence easier to address.

\subsection{Feedback Capacity of Time-Invariant Channel Input Strategies} 
To address Problem~\ref{problem_1}, we restrict  the channel input strategies  to time-invariant, $\Lambda_t=\Lambda^\infty, K_{Z_t}=K_Z^\infty, \forall t$, with corresponding  $X_t=X_t^o, Y_t=Y_t^o, I_t=I_t^o, E_t=E_t^o, K_t=K_t^o\equiv K_t^o(\Lambda^\infty, K_Z^\infty)$. Then we have the following.\\
{\it Generalized time-Invariant DRE:}
\begin{align}
K_{t+1}^o= & c^2 K_{t}^o  + K_{W} -\frac{ \Big( K_{W} + c K_{t}^o\big(\Lambda^\infty + c-a \big)\Big)^2}{ \Big(K_{Z}^\infty+ K_{W} + \big(\Lambda^\infty + c -a\big)^2 K_{t}^o\Big)},  \nonumber\\ 
&K^o_t \geq0, \hso K^o_1=0, \hso t=2,\dots,n.   \label{DRE}
\end{align}
We define feedback capacity $C^{\infty}(\kappa,s)$, as in   \cite{yang-kavcic-tatikonda2007,kim2010}),
\begin{align}
C^{\infty}(\kappa,&s)= \sup_{\stackrel{\Lambda^\infty\in (-\infty, \infty), K_{Z}^\infty\in [0,\infty):}{\lim_{n \longrightarrow \infty} \frac{1}{n} \sum_{t=1}^n \big\{\big(\Lambda^\infty\big)^2 K_{t}^o+ K_{Z}^\infty\big\}   \leq \kappa}} \lim_{n \longrightarrow \infty} \frac{1}{2n} \bigg\{\nonumber
\\  & \sum_{t=1}^n   \log\Big( \frac{\big(\Lambda^\infty+c-a\big)^2 K_{t}^o  + K_{Z}^\infty +K_{W}}{K_{W}}\Big)\bigg\}. \label{IH_ti}
\end{align}
To address  Problem~\ref{problem_1}.(a), we require conditions for convergence  of   DRE  (\ref{DRE}) to  algebraic Riccati equation (ARE),
 \begin{align}
&K^\infty= c^2 K^\infty  + K_{W} \nonumber \\
&-\frac{ \Big( K_{W} + c  K\Big(\Lambda^\infty + c -a\Big)\Big)^2}{ \Big(K_{Z}^\infty+ K_{W} + \Big(\Lambda^\infty + c -a\Big)^2 K^\infty\Big)}, \hst K^\infty\geq 0. \label{DRE_TI_gae} 
\end{align}
 Such conditions require concepts of  detectability and stabilizability  of DREs  \cite{kailath-sayed-hassibi}, as done in  \cite{charalambous-kourtellaris-louka:ARXIV-2020}, using the definitions,
\begin{align} 
&A\tri c, \;\;  C\tri {\Lambda}^\infty + c-a,  \;\; A^* \tri  c-K_W R^{-1}C, \;\; B^{*,\frac{1}{2}} \tri   K_W^{\frac{1}{2}} B^{\frac{1}{2}},\nonumber \\ 
&R \tri  K_Z^\infty + K_W,   \;\;   B \tri  1-K_W\big( K_Z^\infty+K_W\big)^{-1}. \label{notation}
\end{align}
The next notions are directly related to the  asymptotic stability of the error recursion (\ref{i_error_s1}).

\begin{definition}\cite{kailath-sayed-hassibi}\cite{caines1988} Asymptotic stability.\\
A solution $K^\infty\geq 0$ to the generalized ARE (\ref{DRE_TI_gae}), assuming it exists, is called stabilizing if $|F(K^\infty,\Lambda^\infty,K_Z^\infty)|<1$. In this case, we  say $F(K^\infty,\Lambda^\infty,K_Z^\infty)$ is asymptotically stable.
%, that is, $|F(K^\infty,\Lambda^\infty,K_Z^\infty)|<1$. 
\end{definition} 
\begin{definition}\cite{kailath-sayed-hassibi}\cite{caines1988} \\
%Detectability, Stabilizability, Unit Circle controllability.\\
\label{def:det-stab}
(a) The pair  $\big\{A,C\big\}$ is called detectable if there exists a $G \in {\mathbb R}$ such that  $|A- G C|<1$ (stable).\\
(b) The pair $ \big\{A^*, B^{*,\frac{1}{2}}\big\}$ is called unit circle controllable if  there exists a $G \in {\mathbb R}$ such that   $|A^*- B^{*,\frac{1}{2}}G|\neq 1$.\\
(c) The pair $\big\{A^*, B^{*,\frac{1}{2}}\big\}$ is called stabilizable if  there exists a $G \in {\mathbb R}$ such that   $|A^*- B^{*,\frac{1}{2}}G|< 1$.
\end{definition}

In the next theorem we collect known results on the convergence of DREs to AREs.

\begin{theorem}  \cite{kailath-sayed-hassibi} and  \cite{isit2020feedback} \\
\label{thm_ric}
Let  $\{K_t^o, t=1, 2, \ldots, n\}$ denote a sequence that satisfies the DRE (\ref{DRE}) with an arbitrary initial condition.\\ 
Then the following hold.\\
(1) Consider the  DRE (\ref{DRE})  with zero initial condition, i.e., $K_{1}^o=0$, and assume,  the pair $\big\{A,C\big\}$ is detectable, and  the pair $\big\{A^*, B^{*,\frac{1}{2}}\big\}$ is unit circle controllable.\\
The  sequence $\{K_{t}^o: t=1, 2, \ldots, n\}$ that satisfies (\ref{DRE}),  with  $K_{1}^o=0$,  converges, $\lim_{n \longrightarrow \infty} K_{n}^o =K^\infty$, where  $K^\infty$ satisfies   the ARE (\ref{DRE_TI_gae}),
 if and only if the pair $\big\{A^*, B^{*,\frac{1}{2}}\big\}$ is stabilizable.\\
(2) Assume,  the pairs, $\big\{A,C\big\}$ is detectable, and  $\big\{A^*, B^{*,\frac{1}{2}}\big\}$ is unit circle controllable.  Then there exists a unique stabilizing solution $K^\infty\geq 0$ to ARE (\ref{DRE}), i.e.,  such that,  $|F(K^\infty,\Lambda^\infty,K_Z^\infty)|<1$, if and only if  $\{A^*, B^{*,\frac{1}{2}}\}$ is stabilizable.\\
(3) If $\{A, C\}$ is detectable and $\{A^*, B^{*,\frac{1}{2}}\}$ is stabilizable,  then any solution $K_{t}^o, t=1, 2, \ldots,n$ to the  DRE (\ref{DRE})  with arbitrary  initial condition, $K_{1}^o$ is such that $\lim_{n \longrightarrow \infty} K_{n}^o =K^\infty$, where $K^\infty\geq 0$ is the  unique solution of  the generalized ARE (\ref{DRE_TI_gae}) with  $|F(K^\infty, \Lambda^\infty,K_Z^\infty)|<1$, i.e., it is stabilizing.
\end{theorem}
%Moving on, we show the importance of the above theorem. 

\begin{remark} At this point we should emphasize that 
to address  Problem~\ref{problem_1}, we need to impose {\it detectability and stabilizability}. This is fundamentally different from   \cite[Theorem~7,  Corollary~7.1]{yang-kavcic-tatikonda2007} and \cite[Theorem~6.1 and Lemma~6.1]{kim2010} (as explain under literature review), where the {\it  stabilizability condition} is not part of the  optimization problems of  \cite{yang-kavcic-tatikonda2007,kim2010}, 
% that {\it a zero variance of innovations process of the input is optimal} . 
i.e.,  (\ref{ykt_ff})-(\ref{YKT_input_ss}). 
% In particular, 
%if $K_Z^\infty=0$, the two   solutions  of (\ref{DRE_TI_gae}), without the restriction  $K^\infty\geq 0$, are
%\bea
%K^\infty=0, \; K^\infty=\frac{K_W\Big(\big(\Lambda^\infty-a\big)^2-1\Big)}{\Big(\Lambda^\infty+ c-a\Big)^2}, \; \Lambda^\infty +c-a\neq 0. \label{i_rae_nu}
%\eea
%The solution (\ref{i_rae_nu}), that corresponds to  
%$\lim_{n \longrightarrow \infty} K_{n}^o =K^\infty\geq 0$  is determined from   Theorem~\ref{thm_ric}, as follows.  (i) the pair  $\big\{A,C\big\}$ is detectable,  (ii) the pair $ \big\{A^*, B^{*,\frac{1}{2}}\big\}$ is unit circle controllable if and only if $|\Lambda^{\infty}-a| \neq 1$,  and  stabilizable if and only if $|\Lambda^{\infty}-a| < 1$. By Theorem~\ref{thm_ric}.(1), it follows that  $K_{t}^o, t=1, 2, \ldots,n$, with $K_1^o=0$ converges,  $\lim_{n \longrightarrow \infty} K_{n}^o =K^\infty\in [0,\infty)$, if and only if, both  the  unit circle controllability and stabilizability holds,  equivalently  $|\Lambda^{\infty}-a| < 1$. Hence, the unique and stabilizing solution  is $K^\infty=0$.  
Because of this, our answers are   different from\cite{yang-kavcic-tatikonda2007,kim2010}.
\end{remark}

\subsection{Closed-Form Formulas of Feedback Capacity  of  AGN Channels Driven by Nonstationary Noise}
\label{sect:cor-solu}
Using   Theorem~\ref{thm_ric}, we  address Problem~\ref{problem_1} as stated in  Theorem below. 

\begin{theorem} 
 \label{thmf}
Consider the Problem~\ref{problem_1}. Define the set 
\begin{align}
{\cal P}^\infty \tri & \Big\{(\Lambda^\infty, K_Z^\infty)\in (-\infty, \infty)\times [0,\infty): \nonumber \\
& \mbox{(i) the pair  $\{A, C\}$ is detectable,} \nonumber \\
& \mbox{(ii) the pair $\{A^*, B^{*,\frac{1}{2}}\}$}\;\mbox{is stabilizable}\Big\}. \label{adm_set}
\end{align}
Then, 
\begin{align}
&C^\infty(\kappa) \tri \frac{1}{2} \log\Big( \frac{\big(\Lambda^\infty+c-a\big)^2 K^\infty  + K_{Z}^\infty +K_{W}}{K_{W}}\Big) \label{ll_3}\\
&{\cal P}^\infty(\kappa)\tri \Big\{(\Lambda^\infty, K_Z^\infty)\in {\cal P}^\infty: K_Z^\infty\geq 0, 
\ \big(\Lambda^\infty\big)^2 K^\infty + K_{Z}^\infty \leq \kappa\Big\} \nonumber 
%\label{ll_4}
\end{align}
provided there exists $\kappa>0$ such that ${\cal P}^\infty(\kappa)$ is non-empty.  The maximum element $(\Lambda^\infty, K_Z^\infty) \in {\cal P}^\infty(\kappa)$, is such that,  \\(i) if $|c|< 1$, then $(S_t, V_t, X_t, Y_t), t=1, \ldots$ are asymptotic stationary,  and (ii) $(X_t, I_t), t=1, \ldots$ are asymptotic stationary,  $\forall c \in (-\infty,\infty), a \in (-\infty,\infty)$.
\end{theorem} 
 \begin{proof}  By Theorem~\ref{thm_ric} the limits in (\ref{IH_ti}) converge to a unique number and $C^\infty(\kappa,s)$ is independent of $s$. 
 \end{proof}

Theorem \ref{theo:eva} is obtained by solving the optimization problem of Theorem~\ref{thmf}.

\begin{theorem}\label{theo:eva} Consider the optimization problem of Theorem~\ref{thmf}.  Feedback increases capacity for the following regions.\\
$\mbox{A) }c\in\big(1,\sqrt{2}\big)\cup\big(\sqrt{2},\infty\big), a\in\left [ \frac{-c}{c^2-2}, \frac{1}{c} \right ]\label{reg.1}\\
\mbox{B) }c\in\big(-\infty,-\sqrt{2}\big)\cup\big(-\sqrt{2},-1\big), a\in\big(-\infty,\frac{1}{c}\big]\cup \big[\frac{-c}{c^2-2},\infty\big)\label{reg.2}
$
provided the power $\kappa$ satisfies \\
$\kappa > \kappa_{min} \tri \frac{K_W\big(1-ac\big)\big(2ac-ac^3-c^2+ \sqrt{c^3\big(a^2c^3-6ac^2+4a+4c^3-3c\big)}\big)}{2c^2\big(c^2-1\big)^2}$\label{kmin}\\
and the value of  feedback capacity $C^\infty(\kappa)$ is
\begin{align}
&C^\infty(\kappa)= \frac{1}{2}  \log\Big( \frac{\big(\Lambda^{\infty,*}+c-a\big)^2 K^{\infty,*}  + K_{Z}^{\infty,*} +K_{W}}{K_{W}}\Big)\label{sol_1}  \\
& \hspace{.9cm}= \frac{1}{2}  \log\Big( \frac{cK_W(c-2a+a^2c)+c^2\kappa (c^2-1)}{K_W(c^2-1)}\Big),\label{sol_12} \\ 
&K^{\infty,*}= \frac{g}{c(c^2-1)(a-c)^2}, \label{case_11} \\
&\Lambda^{\infty,*}= \frac{K_W(a-c)^2(1-a c)}{g}, \label{case_12}   \\
&K_{Z}^{\infty,*}= \frac{ c \kappa(c^2-1) g-K_W^2(a-c)^2(1-a c)^2}{c(c^2-1)g}.\label{case_13}\\
&g = K_W(2a-c+a^2c^3-2a^2c)+c\kappa(c^2-1)^2
\end{align}
\end{theorem}
\begin{proof} 
The solution is obtained by  a method similar to \cite{ckl2019}.
\end{proof}

The discussion, conclusions and Figures of Section~\ref{sect:results}, related to feedback capacity, are based on Theorem~\ref{theo:eva}.\\
For the  complements of Regimes A and B of Theorem~\ref{theo:eva}, or $\kappa \leq  \kappa_{min}$, there does not exist feedback strategy.  However, we can show that we  can always pick $\Lambda^{\infty}=0$ and ensure a nonfeedback achievable rate.

\subsection{Nonfeedback Achievable Rates  of  IID Channel Input Processes}
Letting $\Lambda^{\infty}=0$, in (\ref{IH_ti}),  the   channel input reduces to  an independent innovation process $X^o_t=Z^o_t,\; t=1,\dots,n$, and hence the code does not use feedback. For such an input  the detectability and stabilizibility conditions are always satisfied, and we obtain  a nonfeedback  achievable rare,  as stated in   the next theorem.

\begin{theorem}
\label{thmnf}
Consider (\ref{IH_ti}), with  $\Lambda^\infty=0$.  and  
An achievable lower bound on nonfeedback capacity is,
\begin{align}
&C_{LB}^{\infty,nfb}(\kappa)\tri \frac{1}{2} \log\Big( \frac{(c-a)^2 K^\infty  + \kappa +K_{W}}{K_{W}}\Big), \hso \forall \kappa \in (0,\infty) \label{nf_888}
\end{align}
where  $K^{\infty}$ is the unique and stabilizable solution  of (\ref{DRE_TI_gae}), with $K_Z^{\infty}=\kappa, \Lambda^\infty=0$, given by 
\begin{align*}
&K^{\infty}= \frac{- h + \sqrt{h^2+4(c-a)^2 K_W \kappa}}{2(c-a)^2} \geq 0\\
&h=\kappa\big(1-c^2) +K_W(1-a^2).
\end{align*}
\end{theorem}
\begin{proof}
This is straightforward to show. 
\end{proof}

%\begin{corollary}
%Achievable rates without feedback for stable and unstable $MA(a)$ and $AR(c)$ noises. \\
%An achievable lower bound for $MA(a)$ noise (\ref{MA}), is
%\begin{align}
%&C_{LB}^{\infty,nfb}\Big|_{\Lambda^\infty=0, c=0}(\kappa)= \frac{1}{2} \log\Big( \frac{a^2 K^{\infty}  + \kappa +K_{W}}{K_{W}}\Big), \;\; \kappa \in [0,\infty) \label{nfma}
%\end{align}
%where $K^{\infty}$ and $K_Z^{\infty}$,  are given by  
%\begin{align}
%&K^{\infty}=\frac{- q + \sqrt{q^2+4a^2 K_W \kappa}}{2a^2} \geq 0,\hso K_Z^{\infty}=\kappa\\
% &q=\kappa+K_W(1-a^2), \hso \Lambda^{\infty}=0, \hso c=0.
%\end{align}
%Achievable rates for the special case of $AR(c)$ noise are given in \cite{isit2020nofeedback}.
%\end{corollary}
%
\section{Acknowledgements}
This work was supported in parts by  the European Regional Development Fund and
 the Republic of Cyprus through the Research Promotion Foundation Projects EXCELLENCE/1216/0365 
and  EXCELLENCE/1216/0296.

\section{Conclusion}\label{sec:con}
In this paper, we characterized and derived closed form   expressions of  feedback capacity  and  achievable  lower bounds on nonfeedback rates, for AGN channels driven by ARMA$(a,c), a \in (-\infty,\infty), c\in (-\infty,\infty),  c \neq a$ noise,  when channel input strategies or distributions are time-invariant. Simulations showed  that the more unstable the noise the higher the feedback capacity, and the  achievable lower bounds on nonfeedback rates.

\newpage

\bibliographystyle{IEEEtran}
\bibliography{Bibliography_capacity}

%\newpage 

%\tableofcontents

\end{document}